\documentclass[aps,pra,twocolumn,superscriptaddress,longbibliography,nofootinbib]{revtex4-1}
\usepackage[normalem]{ulem}

\usepackage{float}
\usepackage{graphicx}
\usepackage{dcolumn}
\usepackage{amssymb}
\usepackage{appendix}
\usepackage{physics}
\usepackage{mathtools}
\usepackage{esvect}
\usepackage{wrapfig}
\usepackage{amsthm}
\usepackage{verbatim}
\usepackage{bbm}
\usepackage[colorlinks=true,linkcolor=blue,citecolor=red,plainpages=false,pdfpagelabels]{hyperref}
\usepackage{subfigure}

\usepackage{placeins}
\usepackage[mathscr]{euscript}

\usepackage{amsmath,amssymb,amsfonts}
\usepackage{algorithmic}
\usepackage{graphicx}
\usepackage{textcomp}
\usepackage[svgnames]{xcolor}
\usepackage{booktabs}
\usepackage{caption}
\usepackage{hyperref}
\usepackage{colortbl}
\usepackage{amsmath}
 \usepackage{bbm}
 \usepackage{braket}
\usepackage[normalem]{ulem}

\usepackage{amsthm, amssymb}

\let\Pr\relax
\definecolor{lightgreen}{HTML}{90EE90}

\def\Pr{\operatorname{Pr}}

\def\({\left(}
\def\){\right)}
\def\[{\left[}
\def\]{\right]}

\newtheorem{theorem}{Theorem}

\newtheorem{proposition}{Proposition}

\def\>{\rangle}
\def\<{\langle}

\begin{document}

\title{Cryptographic Fragility of Standard Quantum Repeater Protocols}

\author{Abhishek Sadhu}
\email{a.sadhu@bham.ac.uk}

\affiliation{School of Computer Science, University of Birmingham, Edgbaston, Birmingham B15 2TT, United Kingdom}

\author{Sharu Theresa Jose}
\affiliation{School of Computer Science, University of Birmingham, Edgbaston, Birmingham B15 2TT, United Kingdom}

\date{\today}

\begin{abstract}
The security of the proposed quantum Internet relies on repeater protocols designed under the assumption of stochastic, characterizable noise. We demonstrate that in adversarial environments this assumption induces performance vulnerabilities for computationally bounded repeater nodes. We show that the standard BBPSSW distillation protocol recursively purifies error syndromes rather than entanglement. This leads to a state of low fidelity despite diagnostic metrics indicating perfect convergence. Moreover, we show that the verifier cannot check the adversarial influence via the maximum likelihood estimation algorithm since it is blind to computationally bounded observers. To address these vulnerabilities, we propose a Cryptographic Network Stack centered on a trapdoor verification protocol. The protocol exploits private randomness to restore operational stability without requiring channel characterization.
\end{abstract}

\maketitle

\section{Introduction}
The realization of a global scale quantum Internet~\cite{Kimble2008,dowling2003quantum} depends on the ability of quantum repeaters to distribute high-fidelity entangled states over long distances~\cite{simon2017towards,sadhu2023practical}. Current architectures are built on the assumption that the noise affecting the quantum links is stochastic, Markovian, and efficiently characterizable~\cite{azuma2023quantum,briegel1998quantum}.  Under these assumptions, the standard protocols such as BBPSSW purification~\cite{bennett1996purification} and maximum likelihood estimation (MLE)~\cite{hradil1997quantum, james2001measurement} are provably robust.

However, in these realistic network environments, the nodes of the network may be computationally bounded to perform only polynomial-time operations. Additionally, the physical layer of a quantum network may be subject to non-stochastic noise, ranging from worst-case environmental noise to active adversarial jamming. 
In quantum networks, it is often assumed that physical layer verification (such as Bell tests) provides sufficient security.

However, the recent work of Leone and co-workers~\cite{Leone2025} introduced a framework of entanglement theory,
where they established that for agents bounded by polynomial-time computation, the thermodynamic limits of entanglement manipulation are governed not by the von Neumann entropy, but by the min-entropy. This, in turn, creates a scenario where an adversary can exploit the gap between von Neumann entropy and min-entropy using states such as pseudoentangled states~\cite{aaronson2022quantum,ji2018pseudorandom} to deceive a computationally bounded verifier.

In this work, we highlight \textit{operational vulnerabilities} imposed by these thermodynamic limits for quantum repeater networks. Specifically, we consider a framework where the channels connecting the repeater nodes may be exposed to an adversary capable of injecting $(\tau,\epsilon)$-pseudoentangled states~\cite{aaronson2022quantum,ji2018pseudorandom}, denoted as $\sigma_{pe}$. Such states~\cite{ghosh2025unconditional} possess low entanglement entropy but are indistinguishable from maximally entangled states to polynomial-time observers. 

We demonstrate that standard entanglement distillation protocols are not useful in such adversarial environments. Specifically, we show that the recurrence map of the BBPSSW protocol~\cite{bennett1996purification} drives pseudoentangled states towards a separable fixed point, while simultaneously satisfying the parity-check statistics of successful purification. After $m$ rounds of the protocol, the repeater observes a sequence of successful parity checks and infers convergence to the target singlet state while the state has effectively collapsed to a separable state with high fidelity. 

We further show that the MLE algorithm~\cite{hradil1997quantum, james2001measurement} is unable to detect structured adversarial noise in polynomial time. An adversary injecting an $\epsilon$-approximate $t$-design state has a polynomial computational cost~\cite{brandao2016local} and can always match any polynomial-bounded verifier with exponentially small distinguishability. This renders the verifier computationally blind to the actions of the adversary. 

The operational failure of the standard protocols is primarily due to their dependence on parameter estimation~\cite{huang2020predicting}. Computing the parameters for structured adversarial noise requires solving BQP-hard problems~\cite{huang2022quantum}. To address these vulnerabilities, we propose a Cryptographic Network Stack centered on a trapdoor verification protocol. The protocol exploits private randomness~\cite{colbeck2011private} to restore operational stability without requiring intractable channel characterization. 

Furthermore, we outline a blind Schur-sampling protocol where the repeater applies the quantum Schur transform~\cite{harrow2005applications,bacon2006efficient} to the buffer and post-selects strictly on the symmetric subspace label. The delocalized structure of the pseudoentangled states~\cite{jeronimo2023pseudorandom} yields an acceptance probability suppressed by the dimension of the symmetric subspace~\cite{harrow2013church}, acting as a super-exponential filter against computational jamming. By filtering interactions based on private randomness and global symmetries, we suggest a physical prerequisite for the robust stability of adversarial quantum links.

The structure of the paper is as follows. In Sec.~\ref{sec:adversary_strategy}, we present an adversarial jamming strategy that leads to the operational failure of standard purification and filtering protocols. We further show that for computationally bounded network nodes, the MLE based tomographic protocols are blind to structured adversarial noisy states. In Sec.~\ref{sec:network_stack}, we propose a Cryptographic Network Stack demonstrating that the vulnerability can be overcome by adopting a cryptographic verification strategy that breaks the symmetry of the injected noise. Finally, we provide concluding remarks and discuss future research directions in Sec.~\ref{sec:discussion}.

\section{The adversarial strategy} \label{sec:adversary_strategy}
Let us consider the BBPSSW protocol \cite{bennett1996purification} applied to a pair of qubits shared between the end nodes. This protocol applies a bilateral CNOT ($U_{\text{CNOT}}^{\otimes 2}$) to two identical input pairs $\rho^{\otimes 2}$,  followed by a local parity measurement. {\color{black} Typically, the input states are assumed to be Bell-diagonal states such as Werner states~\cite{Werner1989}.} 
If the measured parities match,  one of the pairs is preserved, else both are discarded.

We consider a repeater framework where the nodes are honest but computationally bounded, and restricted to polynomial-time operations. Further, the channels connecting the nodes are exposed to an adversary capable of injecting $(\tau, \epsilon)$-pseudoentangled states $\sigma_{pe}$~\cite{aaronson2022quantum,ji2018pseudorandom}. Such states possess low entanglement entropy but are indistinguishable from maximally entangled states to observers restricted to polynomial runtime $\tau$. {\color{black} Specifically, in the BBPSSW purification protocol,  the adversary perturbs the Bell-diagonal state to the following jamming state} 
\begin{align}\rho_{\text{jam}} = (1-\eta)\ket{\Phi^+}\bra{\Phi^+} + \eta~ \sigma_{pe},  \label{eq:jammingstate}\end{align}   where $\sigma_{pe}$ denotes the pseudoentangled state.

\begin{theorem}[Purification Divergence]
Let $\rho^{(m)}_{\text{jam}}$ denote the state after $m$ rounds of the recursive BBPSSW protocol when applied to the adversarial state \eqref{eq:jammingstate}. For any polynomial repeater runtime $m \cdot t_{BBPSSW} \ll \tau$ and noise parameter $\eta \in (0, 1]$, the protocol admits a unique and locally stable fixed point $\sigma_{sep} = \ket{00}\bra{00}$ within the operational subspace. The adversarial state evolves such that
\begin{enumerate}
    \item The parity check success probability converges to unity, i.e., $\lim_{m\rightarrow\infty}p_{succ}^{(m)}=1$.
    \item The fidelity converges to the separable noise floor, i.e., $\lim_{m\rightarrow\infty}F(\rho^{(m)}_{\text{jam}})=1/2$.
\end{enumerate}
\end{theorem}

\textit{Proof Sketch}--- The target state $\ket{\Phi^+} = (\ket{00}+\ket{11})/\sqrt{2}$ resides in the even-parity subspace of the total Hilbert space. Over multiple rounds of the protocol, the singlet component $\ket{\Phi^+}$ suffers attenuation due to channel noise thereby decreasing its population. However, the adversarial component $\sigma_{\text{sep}}$,  which is hidden within $\sigma_{pe}$, is preserved. The purification protocol therefore acts as an amplifier for the separable component relative to the entangled signal. After $m$ rounds, the repeater observes a sequence of successful parity checks and infers convergence to $\ket{\Phi^+}$, while the state has effectively collapsed to $\sigma_{\text{sep}} = \ket{00}\bra{00}$. Note that the separable state $\ket{00}$ is also an eigenstate of the parity operator with eigenvalue $+1$. Consequently, $\sigma_{\text{sep}}$ passes the BBPSSW filter with probability 1. The map is designed to filter odd-parity error states (like $\ket{01}$) but remains transparent to the $\ket{00}$ error states. This represents a complete decoupling of the diagnostic metric ($p_{\text{succ}}$) from the physical resource (entanglement). We present more details in Appendix~\ref{app:BBPSSW_map}. 

This vulnerability extends beyond multi-copy protocols to single-copy protocols such as \textit{entanglement filtering}~\cite{bennett1996concentrating,horodecki2009quantum}. When the adversary injects the pseudoentangled state, the local filtering operation suffers from stagnation. Specifically, the local marginals of the pseudoentangled noise $\sigma_{pe}$ are computationally indistinguishable from the maximally mixed state to polynomial-time observers. Hence, the filtering algorithm evaluates the local marginals as maximally mixed. Consequently, the iterative map stagnates and freezes the fidelity at the injected noisy state despite diagnostic metrics indicating a successfully distilled singlet. We present a proof of this in Appendix~\ref{app:filter_map}.

These structural vulnerabilities highlight that the process of distillation can be exploited by the adversary to stabilize low-$S_{\min}$ states. 
To detect such failures, networks typically employ a heartbeat verification using the MLE algorithm~\cite{james2001measurement, hradil1997quantum, smolin2012efficient}. We show that such protocols are incapable of detecting structured adversarial noise in polynomial time.

\begin{proposition}[Tomographic Blindness]
Let $\mathcal{V}$ be any tomographic verification protocol utilizing an arbitrary POVM $\mathcal{M} = \{M_i\}$ on a $k$-copy buffer to estimate the channel noise. If the adversary injects an ensemble $\mu_{adv}$ forming an $\epsilon$-approximate $t$-design on $\mathcal{H} \cong \mathbb{C}^d$ ($d=2^n$), the maximum statistical advantage of distinguishing $\mu_{adv}$ from the ideal Haar-random noise ensemble $\mu_{Haar}$ is bounded by
\begin{equation}
    \text{Adv}_{\mathcal{V}} \le \frac{1}{2} \left\| \rho_{\text{jam}}^{(k)} - \rho_{Haar}^{(k)} \right\|_1 \le \epsilon \quad \text{for all } k \le t.
\end{equation}
Since generating such an ensemble requires $\mathcal{O}(t \cdot \text{poly}(n))$ computational cost, a polynomial-time adversary can always match any polynomial-memory verifier ($k \in \text{poly}(n)$) with exponentially small distinguishability $\epsilon = \mathcal{O}(2^{-n})$, thereby rendering the verifier computationally blind.
\end{proposition}

\textit{Proof Sketch}---  A quantum $t$-design matches the statistical moments of the Haar measure up to order $t$. Distinguishing the adversarial ensemble from the ideal Haar-random ensemble requires resolving the collision probability which scales as $2^{-n/2}$. For a repeater with polynomial memory $k \ll 2^{n/2}$, the trace distance between the observed statistics is exponentially small. Thus, the adversary is invisible not due to instrumental imperfection, but due to the sample complexity limits of the Hilbert space. The detailed proof is presented in Appendix~\ref{app:tomographic_blindness}.

As previously discussed, the operational failure of standard protocols is due to their dependence on parameter estimation. Under structured adversarial noise, determining the noise parameters is a hard problem. We next present a heuristic showing that this fragility can be overcome by employing cryptographic verification strategy that breaks the symmetry of the injected noise. 

\section{Cryptographic Network Stack} \label{sec:network_stack}
We first introduce a trapdoor verification protocol. The network controller distributes a secret seed $s\in\{0,1\}^{\lambda}$ to drive a quantum-secure pseudorandom function (PRF). The PRF determines a sequence of measurement bases as $B_{i}=\text{PRF}(s,i)$ for the $i$-th qubit pair in a $k$-copy buffer. The verifier measures the buffered state and accepts only if the correlations violate the Bell-CHSH inequality~\cite{Clauser1969,sadhu2023testing}. Since the adversary injecting the pseudorandom state $\rho_{\text{jam}}$ lacks access to the PRF seed, they are forced to contend with the average-case hardness of the verification.

\begin{theorem}[Trapdoor Security]
Let $\text{PRF}(s,i)$ be a secure pseudorandom function. For any computationally bounded adversary lacking the seeds, the probability that a pseudoentangled state $\sigma_{pe}$ passes the trapdoor verification check is bounded by
\begin{equation}
    Pr[Pass]\le2^{-\Omega(k)}+negl(\lambda).
\end{equation}
\end{theorem}
The formal cryptographic reduction is detailed in Appendix~\ref{app:trapdoor}. 

While the trapdoor protocol secures the network through computational unpredictability, a robust cryptographic stack must also exploit the physical structure of the noise. Beyond private randomness, repeaters can operationalize the extraction of min-entropy via global symmetry filtering. As a complementary layer of security, we outline a blind Schur-sampling protocol where the repeater applies the quantum Schur transform $U_{Schur}$ \cite{harrow2005applications, bacon2006efficient} to the buffer and post-selects strictly on the symmetric subspace label $\lambda_{sym}=(k,0,\dots,0)$.
While the ideal state $|\Phi^{+}\rangle^{\otimes k}$ is accepted with probability 1, the delocalized structure of the pseudoentangled noise yields an acceptance probability $P \approx 1/k!$ suppressed by the dimension of the symmetric subspace. This acts as a super-exponential filter against computational jamming. Further details are provided in Appendix~\ref{app:schur}.

\section{Discussion} \label{sec:discussion}
In this work, we analysed the stability of standard quantum repeaters under structured adversarial noise. Our analysis indicates that entanglement purification and tomographic verification can be subverted by pseudoentangled states as inputs  which approximate $t$-designs when processed by polynomial-time bounded repeaters. In these cases, conventional diagnostic metrics may indicate successful operation even though the underlying entanglement is compromised.

Specifically, we show that the BBPSSW recurrence map can amplify the separable components of pseudoentangled states. The map admits a separable fixed point mimicking the parity-check statistics of an ideal singlet. Furthermore, our results show that tomography-based verification can be computationally blind to adversarial ensembles forming approximate $t$-designs. For verifiers restricted to polynomial memory and runtime, this structured noise becomes statistically indistinguishable from stochastic noise. 

These observations highlight the potential need to integrate cryptographic elements into the quantum network stack. To address these vulnerabilities, we outline two protocols. First, we present a trapdoor verification which uses private randomness to prevent pseudoentangled input states from entering the buffer. Second, we present a blind Schur-sampling which leverages permutation symmetry to filter out structured noise. We view these protocols as initial steps toward a more resilient network design. For future work, it will be interesting to observe the trade-offs between cryptographic robustness and resource overhead, alongside investigating how these concepts apply to large-scale, fault-tolerant repeater architectures. 

\section{Acknowledgements}
AS and STJ received funding from EPSRC Quantum Technologies Career Acceleration Fellowship (UKRI1218).

\bibliography{ref}

\appendix

\section{Analysis of the BBPSSW Map Under Adversarial Perturbation} \label{app:BBPSSW_map}

In this section, we analyze the dynamics of the BBPSSW purification map $\Gamma$ \cite{bennett1996purification} acting on the adversarially produced state $\rho_{\text{jam}}$. Unlike standard convergence proofs which assume that the input state resides in the Bell-diagonal subspace (e.g., Werner states \cite{Werner1989}), we drop this assumption to account for the structured correlations of pseudoentanglement. Specifically, we model the adversarial jamming state $\rho_{\text{jam}}$ defined on the Hilbert space $\mathcal{H}_{AB} = \mathbb{C}^2 \otimes \mathbb{C}^2$ as
\begin{equation}
\rho_{\text{jam}} = (1 - \eta) |\Phi^+\rangle\langle\Phi^+| + \eta~ \sigma_{pe}
\end{equation}
where $\sigma_{pe}$ is $(t, \epsilon)$-pseudoentangled \cite{aaronson2022quantum, ji2018pseudorandom} i.e., for any quantum channel $\mathcal{C}$ implementable by a quantum circuit of size (or runtime) at most $t$, $||\mathcal{C}(\sigma_{pe}) - \mathcal{C}(\sigma_{sep})||_1 \le \epsilon$, where $\sigma_{sep}$ is a separable product state. Without loss of generality, we set $\sigma_{sep} = |00\rangle\langle00|$.

The BBPSSW protocol consists exclusively of polynomial-time operations \cite{bennett1996purification}. The protocol acts on two copies $\rho_{\text{jam}}^{\otimes 2}$ living in $\mathcal{H}_{A_1 B_1} \otimes \mathcal{H}_{A_2 B_2}$, where $A_1, B_1$ are source qubits and $A_2, B_2$ are target qubits. The protocol first applies a bilateral CNOT operation $U_{CNOT}^{\otimes 2} = U_{A_1 \to A_2} \otimes U_{B_1 \to B_2}$ on the input state, followed by a local parity measurement on the target qubit. The BBPSSW purification map $\Gamma(\cdot)$ is defined via the projection operator onto the success branch given as
\begin{equation}
P_{succ} = \mathbb{I}_{A_1 B_1} \otimes |00\rangle\langle00|_{A_2 B_2} + \mathbb{I}_{A_1 B_1} \otimes |11\rangle\langle11|_{A_2 B_2}
\end{equation}

To analyze the multi-round dynamics, we first bound the propagation of the computational indistinguishability under the protocol's post-selection step. Since the parity check succeeds with probability $p_{succ}=1$ in the even-parity subspace, simulating the accepted branches does not require an exponential overhead. After $m$ rounds, the accumulated error is $\mathcal{O}(m\epsilon)$. Provided the cumulative runtime satisfies $m \cdot t_{BBPSSW} \ll t$, the state $\sigma_{pe}$ remains computationally indistinguishable from $\sigma_{sep}$.

With this bound established, we parameterize the effective state at an arbitrary round $m$ as $\rho^{(m)} = x_m |\Phi^+\rangle\langle\Phi^+| + y_m |00\rangle\langle00| + z_m |11\rangle\langle11|$, where $x_m + y_m + z_m = 1$. The initial state at $m=0$ is $x_0 = 1 - \eta$, $y_0 = \eta$, and $z_0 = 0$.

Applying the bilateral CNOT $U_{CNOT}^{\otimes 2}$ and the target parity projection $P_{succ}$ to $\rho^{(m)} \otimes \rho^{(m)}$ yields the unnormalized coefficients for round $m+1$. Because the target parity is identically even for all tensor products of these basis states, there is no rejection, yielding a total success probability $p_{succ}^{(m)} = 1$. Tracing over the target qubits, the normalized recurrence map is given by:
\begin{eqnarray}
\tilde{x}_{m+1} &=& x_m^2 \\
\tilde{y}_{m+1} &=& y_m^2 + \frac{3}{2}x_m y_m + y_m x_m + \frac{1}{2}x_m z_m + y_m z_m \nonumber \\ \\
\tilde{z}_{m+1} &=& z_m^2 + \frac{1}{2}x_m y_m + \frac{3}{2}x_m z_m + z_m x_m + z_m y_m \nonumber \\
\end{eqnarray}

By substituting $z_m = 1 - x_m - y_m$, the quadratic terms cancel, and the map simplifies exactly to:
\begin{eqnarray}
x_{m+1} &=& x_m^2 \\
y_{m+1} &=& y_m + \frac{1}{2}x_m(1 - x_m) \\
z_{m+1} &=& z_m + \frac{1}{2}x_m(1 - x_m)
\end{eqnarray}

To rigorously evaluate the stability of the adversarial fixed point, we compute the Jacobian $J$ of the map restricted to the $(x_m, y_m)$ parameters~\cite{strogatz2001nonlinear} as
\begin{equation}
J = \begin{pmatrix} 2x_m & 0 \\ \frac{1}{2} - x_m & 1 \end{pmatrix}
\end{equation}

Evaluated at the separable boundary ($x_m = 0$), the eigenvalues are $\lambda_1 = 0$ and $\lambda_2 = 1$. The eigenvalue of $0$ dictates that the entangled singlet fraction ($x_m$) collapses quadratically to zero, driving the system deterministically onto the invariant separable manifold $y+z=1$. Crucially, standard BBPSSW diagnostics monitor only $p_{succ}$. Since $p_{succ} \to 1$ identically for this state, the repeater cannot distinguish convergence to the ideal singlet $|\Phi^+\rangle$ from convergence to the separable manifold, while the true state collapses to fidelity $F = 1/2$. This proves that $\Gamma$ is dynamically unstable against adversarial initialization.

\section{Analysis of entanglement filtering Under Adversarial Perturbation} \label{app:filter_map}

Consider a repeater framework where the nodes are honest but computationally bounded, and restricted to polynomial-time operations~\cite{Leone2025}. The adversary injects the bipartite jamming state defined on the Hilbert space $\mathcal{H}_{AB} = \mathbb{C}^2 \otimes \mathbb{C}^2$ as $\rho_{\text{jam}}^{(0)} = (1-\eta)|\Phi^+\rangle\langle\Phi^+| + \eta~\sigma_{pe}$ where $\sigma_{pe}$ denotes the $(\tau, \epsilon)$-pseudoentangled state~\cite{aaronson2022quantum}. At any arbitrary round $m \ge 1$, the repeater nodes attempt to extract maximal entanglement by applying generalized local filtering operations~\cite{bennett1996concentrating} denoted by the POVM elements $M_A^{(m)}$ and $M_B^{(m)}$. One of the nodes, say Alice traces out the other node, say Bob's subsystem to determine her reduced density matrix based on the prior round as $\rho_A^{(m)} = \text{Tr}_B \rho^{(m-1)}_{\text{jam}}$.

To apply the filtering operation, Alice first diagonalizes the state via a local unitary $V_A^{(m)}$ to get the rotated state $\tilde{\rho}_A^{(m)} = V_A^{(m)} \rho_A^{(m)} V_A^{(m)\dagger}$. Note that the state $\tilde{\rho}_A^{(m)}$ is diagonal as $\tilde{\rho}_A^{(m)} = \lambda_0|0\rangle\langle0| + \lambda_1|1\rangle\langle1|$,  where $\lambda_0$ and $\lambda_1$ are the eigenvalues of $\rho_A^{(m)}$ satisfying $\lambda_0 + \lambda_1 = 1$.

In this diagonalized basis, Alice defines the filter $F_A^{(m)} = \text{diag}(f_0, f_1)$. The action of this filter on the state $\tilde{\rho}_A^{(m)}$ yields the unnormalized state $F_A^{(m)} \tilde{\rho}_A^{(m)} F_A^{(m)\dagger} = f_0^2 \lambda_0|0\rangle\langle0| + f_1^2 \lambda_1|1\rangle\langle1|$.

To distill a maximally entangled state, the filtered local marginal must be driven toward the maximally mixed state $I/2$. This requires the amplitudes to be perfectly balanced, yielding $f_0^2 \lambda_0 = f_1^2 \lambda_1$. To optimize the probability of successful filtration while satisfying the physical POVM constraint, Alice actively suppresses the higher-eigenvalue state resulting in the following choice of filter parameters
\begin{equation}
f_0 = \min\left(1, \sqrt{\frac{\lambda_1}{\lambda_0}}\right), \quad f_1 = \min\left(1, \sqrt{\frac{\lambda_0}{\lambda_1}}\right).
\end{equation}

Then, the composition of the basis transformation, the diagonal filter, and the inverse basis transformation defines the complete local POVM operation of Alice $M_A^{(m)}$ defined as
\begin{equation}
M_A^{(m)} = V_A^{(m)\dagger} F_A^{(m)} V_A^{(m)}.
\end{equation}

Now, to execute the first round of filtering ($m=1$) Alice evaluates her local marginal of the initial adversarial state as
\begin{equation}
\rho_A^{(0)} = (1-\eta)~\text{Tr}_B ( |\Phi^+\rangle\langle\Phi^+| ) + \eta~\text{Tr}_B ( \sigma_{pe} ).
\end{equation}

The ideal singlet fraction yields a maximally mixed marginal $\text{Tr}_B(|\Phi^+\rangle\langle\Phi^+|) = \mathbb{I}/2$. We next evaluate $\text{Tr}_B(\sigma_{pe})$. By definition, the state $\sigma_{pe}$ has low entanglement entropy and are indistinguishable from the maximally entangled states to polynomial-time verifier. This renders the verifier computationally blind and we have $\text{Tr}_B(\sigma_{pe}) \approx_c \mathbb{I}/2$. Consequently, Alice's local tomography shows maximally mixed state. 

Now since the local state $\rho_A^{(0)}$ is already proportional to the identity matrix, it is invariant under all unitary transformations. Any arbitrary local unitary $V_A^{(1)}$ chosen by the algorithm will trivially diagonalize it, yielding eigenvalues $\lambda_0 = \lambda_1 = 1/2$. Inserting the eigenvalues into the filter function expressions yields $f_0 = f_1 = 1$ and generates the diagonal filter $F_A^{(1)} = \mathbb{I}$. By symmetry, Bob also computes $M_B^{(1)} = \mathbb{I}$. Applying these generalized operations yields the normalized state for the first round of the protocol as $\rho^{(1)}_{\text{jam}} = \rho^{(0)}_{\text{jam}}$.

By mathematical induction, let at any arbitrary round $k$ the state remains $\rho^{(k)}_{\text{jam}} = \rho^{(0)}_{\text{jam}}$. To compute the operations for round $k+1$, the nodes perform tomography on $\rho^{(k)}_{\text{jam}}$. It then follows that the state is $\text{Tr}_B(\rho^{(k)}_{\text{jam}}) = \text{Tr}_B(\rho^{(0)}_{\text{jam}}) \approx_c \mathbb{I}/2$. The filtering algorithm computes eigenvalues $\lambda_0 = \lambda_1 = 1/2$ and calculates $F_A^{(k+1)} = F_B^{(k+1)} = \mathbb{I}$. Thus, the state updates to $\rho^{(k+1)}_{\text{jam}} = \rho^{(k)}_{\text{jam}} = \rho^{(0)}_{\text{jam}}$. For any arbitrarily high number of filtering iterations, we have
\begin{equation}
\lim_{m \rightarrow \infty} \rho^{(m)}_{\text{jam}} = \rho^{(0)}_{\text{jam}}.
\end{equation}

Thus, we observe that the filtering operation is blind to the adversarially injected state and converges to a state with low fidelity despite diagnostic metrics indicating perfect convergence. 

\section{Measure-Theoretic Analysis of Tomographic Blindness} \label{app:tomographic_blindness}

We formally analyze the security of maximum likelihood estimation (MLE) algorithm and general tomographic verification\cite{hradil1997quantum, james2001measurement, smolin2012efficient} as a quantum state discrimination problem~\cite{bae2015quantum, chefles2000quantum} between two measures on the Hilbert space $\mathcal{H}\cong\mathbb{C}^d$ ($d=2^n$): the ideal Haar measure $\mu_{Haar}$ representing stochastic channel noise, and the adversarial measure $\mu_{adv}$.

The verification primitive $\mathcal{V}$ acts on a buffer of $k$ copies subject to local network constraints \cite{bergou2010discrimination}. The state of the buffer is defined by the $k$-th moment of the respective ensembles as
\begin{equation}
    \rho_{\mu}^{(k)}=\int (|\psi\rangle\langle\psi|)^{\otimes k}d\mu(\psi).
\end{equation}
The optimal state discrimination strategy between these two ensembles using any generalized measurement (POVM) $\mathcal{M} = \{M_i\}$ is strictly governed by the Helstrom bound~\cite{bergou200411}. The probability of correctly detecting the adversarial ensemble, assuming equal prior probabilities, is given by $P_{detect} \le \frac{1}{2} + \frac{1}{4} \|\rho_{\text{jam}}^{(k)} - \rho_{Haar}^{(k)}\|_1$. The distinguishing advantage is therefore bounded by the trace distance $\Delta = \frac{1}{2} \|\rho_{\text{jam}}^{(k)} - \rho_{Haar}^{(k)}\|_1$.

We assume the adversary is bounded by quantum polynomial time (BQP), capable of executing circuits of depth $D = \mathcal{O}(\text{poly}(n))$. Using pseudo-random quantum state (PRS) constructions~\cite{ji2018pseudorandom}, the adversary can efficiently generate an ensemble $\mu_{adv}$ that forms an $\epsilon$-approximate $t$-design~\cite{brandao2016local}. By definition, the $k$-th moments of an $\epsilon$-approximate $t$-design match those of the Haar measure up to an error $\epsilon$ in the trace norm for all $k \le t$ as
\begin{equation}
    \left\| \rho_{\text{jam}}^{(k)} - \rho_{Haar}^{(k)} \right\|_1 \le 2\epsilon \quad \text{for } k \le t.
\end{equation}
To evade a verifier utilizing a buffer of size $k$, the adversary simply parameterizes their generation circuit such that $t \ge k$. The computational cost to generate a unitary $t$-design scales as $\mathcal{O}(t \cdot poly(n))$ \cite{brandao2016local}. 

Crucially, physical repeaters are restricted to polynomial quantum memory, meaning $k \in \mathcal{O}(\text{poly}(n))$. Consequently, the adversary's required circuit depth $D = \mathcal{O}(k \cdot \text{poly}(n))$ remains strictly within polynomial bounds. Furthermore, modern cryptographic PRS constructions allow the adversary to achieve an exponentially small approximation error $\epsilon = \text{negl}(n) = \mathcal{O}(2^{-n})$ while maintaining this polynomial gate overhead.

Therefore, for any estimator or POVM $\mathcal{M}$ executed by the polynomial-time verifier, the total variation distance in the observed statistics is bounded by
\begin{equation}
    \max_{\mathcal{M}} \frac{1}{2} \sum_i \left| \text{Tr}(M_i \rho_{\text{jam}}^{(k)}) - \text{Tr}(M_i \rho_{Haar}^{(k)}) \right| \le \epsilon.
\end{equation}
The verifier's maximum distinguishing advantage is bounded by $\mathcal{O}(2^{-n})$. The adversary is thus invisible to tomographic verification not due to instrumental imperfection, but because the statistical distance between the ensembles is information-theoretically erased within the computational limits of the repeater.

\section{Security Reduction for Trapdoor Verification} \label{app:trapdoor}

We provide a cryptographic reduction for the security of the trapdoor verification protocol. We model the adversary $\mathcal{A}$ as a {\color{black} quantum polynomial-time} (QPT) algorithm attempting to win a distinguishing game against the repeater $\mathcal{R}$.

\textbf{Game $\mathcal{G}$:} The repeater $\mathcal{R}$ and the adversary $\mathcal{A}$ plays the following game: 
\begin{enumerate}
\item[Step 1] $\mathcal{R}$ samples a key $K \leftarrow \{0,1\}^\lambda$ and uses a quantum-secure pseudorandom function (PRF) $F_K$ to derive measurement bases $B_i = F_K(i)$ for {\color{black} each} $i \in \{1, \dots, k\}$, {\color{black} where $i$ indexes the $k-$th qubit pairs in the buffer.} 
\item[Step 2]  $\mathcal{A}$ submits a state $\rho$ to the buffer.
\item[Step 3]  $\mathcal{R}$ measures $\rho$ in bases $\{B_i\}$. 
\end{enumerate}
$\mathcal{A}$ wins if the results violate the CHSH inequality consistent with a Bell pair.

\textbf{Theorem.} If $F_K$ is a secure PRF, then for any adversary $\mathcal{A}$ inserting a separable state, the winning probability is $\text{Pr}[\text{Win}] \le 2^{-\Omega(k)} + \text{negl}(\lambda)$.

\textit{Proof.} We proceed by a hybrid argument. Let $\mathcal{G}_{\text{real}}$ be the game where bases are derived from $F_K$. Let $\mathcal{G}_{\text{ideal}}$ be the game where bases $B_i$ are sampled truly uniformly at random from the Haar measure. By the security definition of a PRF, $|\text{Pr}[\text{Win}_{\text{real}}] - \text{Pr}[\text{Win}_{\text{ideal}}]| \le \text{Adv}_{\text{PRF}}$ which is negligible for secure PRFs. {\color{black} Here, the PRF advantage $\text{Adv}_{\text{PRF}}$ is defined as the absolute difference in the probability that $\mathcal{A}$ outputs 1 (success) when given oracle access to the PRF $F_K$ as compared to a truly random function $f$. This can be written as $\text{Adv}_{\text{PRF}} = \left| \Pr[A^{F_K(\cdot)} = 1] - \Pr[A^{f(\cdot)} = 1] \right|$.}
 
{\color{black} In $\mathcal{G}_{\text{ideal}}$, the chosen bases are uniformly random and are unknown to the adversary $\mathcal{A}$. Consider the CHSH observable $\mathcal{C}$. For a fixed separable state $\sigma_{\text{sep}}$, the correlations are strictly bounded by $\langle \mathcal{C} \rangle \le 2$, For a singlet, $\langle \mathcal{C} \rangle = 2\sqrt{2}$. Let $V_i$ be the binary variable indicating if the $i$-th qubit pair successfully passes the CHSH check. Now, as the verifier's basis choice $B_i$ is random, and the adversarial state must be fixed independently of this choice, the adversary cannot deterministically align the state $\sigma_{\text{sep}}$ to spoof the measurement. Consequently, the expectation value of passing the check for any single pair which is defined as the maximum success probability p is therefore strictly less than 1. This corresponds to $\mathbb{E}[V_i] \le p < 1$.}  

The total probability of passing $k$ independent checks is $p^k = 2^{-\Omega(k)}$. {\color{black} Consequently, the overall success probability is bounded by $\text{Pr}[\text{Win}_{\text{real}}] \leq \text{Pr}[\text{Win}_{\text{ideal}}] + \text{Adv}_{\text{PRF} } \leq 2^{-\Omega(k)}+ neg(\lambda)$.} Thus, the use of private randomness forces the adversary to contend with the average-case hardness of the verification, rather than optimizing for a specific, known basis set. 

\section{Security Reduction for Blind Schur-Sampling}~\label{app:schur}

In this section, we present the proposed \textit{blind Schur-sampling} protocol as an alternative implementation of the theoretical limits derived by Leone et al.~\cite{Leone2025}. The protocol is described as the following game $\mathcal{G}_{\text{Schur}}$ between the repeater $\mathcal{R}$ and the adversary $\mathcal{A}$: \begin{enumerate} \item[Step 1] The repeater $\mathcal{R}$ initializes a buffer of size $k = \text{poly}(n)$.
\item [Step 2] The adversary $\mathcal{A}$ submits a state {\color{black} $\rho^{\otimes k}_{\rm jam}$} to the buffer.
\item [Step 3] \textbf{Verification:} $\mathcal{R}$ applies the Quantum Schur Transform (QST) $U_{\text{Schur}}$ on the input state and measures the irreducible representation label $\lambda$.
\item[Step 4] \textbf{Verdict:} $\mathcal{R}$ accepts the state if $\lambda = \lambda_{\text{sym}} \equiv (k, 0, \dots, 0)$, corresponding to the symmetric subspace. Otherwise, $\mathcal{R}$ rejects.
\end{enumerate}
\textbf{Theorem.} For any adversary $\mathcal{A}$ submitting a pseudoentangled state $\sigma_{pe}$ (computationally indistinguishable from a Haar-random state), the acceptance probability is bounded by $\text{Pr}[\text{Accept}] \le \frac{1}{k!} + \text{negl}(n)$.

\textit{Proof.} We analyze the acceptance condition for the ideal and adversarial cases. Let $\Pi_{\text{sym}}$ be the projector onto the symmetric subspace of $(\mathbb{C}^d)^{\otimes k}$. The acceptance probability is $P_{\text{acc}} = \text{Tr}(\Pi_{\text{sym}} \rho^{\otimes k})$.

\textit{Case 1: Ideal Signal.} If the input is the target state $\ket{\Phi^+}^{\otimes k}$, the state is invariant under all permutations of the $k$ copies. Thus, $\rho^{\otimes k}$ resides entirely within the symmetric subspace, and $P_{\text{acc}}(\Phi^+) = 1$.

\textit{Case 2: Adversarial Noise.} The pseudoentangled state $\sigma_{pe}$ is indistinguishable from a Haar-random state $\phi_{\text{Haar}}$. By Levy's Lemma, a random state in the tensor product space $(\mathbb{C}^d)^{\otimes k}$ is maximally delocalized across the Schur sectors. The dimension of the symmetric subspace is $D_{\text{sym}} = \binom{d+k-1}{k}$~\cite{harrow2013church} while the total dimension is $D_{\text{total}} = d^k$. The probability of a random state projecting onto $\lambda_{\text{sym}}$ is given by the ratio of dimensions:
\begin{equation}
    P_{\text{acc}}(\phi_{\text{Haar}}) = \frac{D_{\text{sym}}}{D_{\text{total}}} = \frac{(d+k-1)!}{k!(d-1)!d^k} \approx \frac{1}{k!}.
\end{equation}
Since $\sigma_{pe} \approx_c \phi_{\text{Haar}}$, the acceptance probability for the adversarial state is bounded by this ratio up to a negligible factor. For any buffer size $k > 2$, this probability vanishes super-exponentially. Thus, the blind Schur-sampling protocol acts as a cryptographic filter with extinction ratio $R \approx k!$, rejecting adversarial states without requiring parameter estimation.

\end{document}